\documentclass[numberwithinsect,a4paper,UKenglish,cleveref, autoref, thm-restate]{lipics-v2019}

\title{Red-Blue Point Separation for Points on a Circle} 

\titlerunning{Red-Blue Point Separation for Points on a Circle} 

\author{Neeldhara Misra}{Indian Institute of Technology, Gandhinagar, India}{neeldhara.m@iitgn.ac.in}{}{}
\author{Harshil Mittal}{Indian Institute of Technology, Gandhinagar, India}{mittal\_harshil@iitgn.ac.in}{}{}
\author{Aditi Sethia}{Indian Institute of Technology, Gandhinagar, India}{aditi.sethia@iitgn.ac.in}{}{}

\authorrunning{N. Misra and H. Mittal and A. Sethia} 

\Copyright{Neeldhara Misra and Harshil Mittal and Aditi Sethia} 

\ccsdesc[100]{Theory of computation~Design and analysis of algorithms} 

\keywords{red-blue point separation, axis-parallel lines, circle} 

\category{} 

\relatedversion{} 

\supplement{}

\EventEditors{John Q. Open and Joan R. Access}
\EventNoEds{2}
\EventLongTitle{42nd Conference on Very Important Topics (CVIT 2016)}
\EventShortTitle{CVIT 2016}
\EventAcronym{CVIT}
\EventYear{2016}
\EventDate{December 24--27, 2016}
\EventLocation{Little Whinging, United Kingdom}
\EventLogo{}
\SeriesVolume{42}
\ArticleNo{23}
\usepackage{amssymb,amsmath}
\usepackage[svgnames]{xcolor}
\usepackage{hyperref}
\hypersetup{
    colorlinks,
    linkcolor={IndianRed},
    citecolor={SteelBlue},
    urlcolor={SeaGreen}
}

\usepackage[OT1]{fontenc}
\usepackage[utf8]{inputenc}
\usepackage{charter,eulervm}
\usepackage{parskip}
\usepackage{tikz}
\usetikzlibrary{snakes,decorations.pathreplacing,calc}
\hideLIPIcs{}
\nolinenumbers

\begin{document}

\maketitle

\begin{abstract}
    Given a set $R$ of red points and a set $B$ of blue points in the plane, the Red-Blue point separation problem asks if there are at most $k$ lines that separate $R$ from $B$, that is, each cell induced by the lines of the solution is either empty or monochromatic (containing points of only one color). A common variant of the problem is when the lines are required to be axis-parallel. The problem is known to be NP-complete for both scenarios~\cite{CalinescuDKW05,Megiddo88}, and \textsf{W[1]}-hard parameterized by $k$ in the former setting~\cite{BonnetGL17} and \textsf{FPT} in the latter~\cite{KratschMMPS2020}. We demonstrate a polynomial time algorithm for the special case when the points lie on a circle.  Further, we also demonstrate the W-hardness of a related problem in the axis-parallel setting, where the question is if there are $p$ horizontal and $q$ vertical lines that separate $R$ from $B$. The hardness here is shown in the parameter $p$.
\end{abstract}

\section{Introduction}

Given a set $R$ of red points and a set $B$ of blue points in the plane, the \textsc{Red-Blue Separation} (\textsf{RBS}) problem asks if there are at most $k$ lines that separate $R$ from $B$, that is, each cell induced by the lines of the solution is either empty or monochromatic (containing points of only one color). Equivalently, $R$ is separated from $B$ if, for every straight-line segment $\ell$ with one endpoint in $R$ and the other one in $B$, there is at least one line in the solution that intersects $\ell$.  A common variant of the problem is when the solution lines are required to be axis-parallel (\textsf{APRBS}). Questions about the discrete geometry on red and blue points in general, and their separability using geometric objects in particular, are of fundamental interest. This makes \textsf{RBS} a well-studied problem on its own right. It is also motivated by the problem of univariate discretization of continuous variables in the context of machine learning~\cite{FayyadI93,KujalaE07}. 



\textsf{RBS} is known to be \textsf{NP}-complete~\cite{Megiddo88}, \textsf{APX}-hard~\cite{CalinescuDKW05}, and \textsf{W[1]}-hard when parameterized by the solution size~\cite{BonnetGL17}. The approximation hardness holds for the \textsf{APRBS} problem as well, while in contrast the parameterized intractability applies only to the \textsf{RBS} problem. Specifically, it is known that an algorithm running in time $f(k)n^{o(k/\log k)}$, for any computable function $f$, would disprove ETH~\cite{BonnetGL17}. This reduction crucially relies on selecting lines from a set with a large number of different slopes --- in particular, the number of distinct slopes of the lines used is not bounded by a function of~$k$. In contrast, a recent development demonstrates that \textsf{APRBS} is, in fact, \textsf{FPT} parameterized by the solution size~\cite{KratschMMPS2020}. Prior to this, the best that was known in this context was an algorithm from~\cite{BonnetGL17} showing that \textsf{APRBS} is \textsf{FPT} parameterized by the number of blue points (or the number of red points). 

For the case where $k = 1$ and $k = 2$, the problem is solvable in $O(n)$ and $O(n\log n)$ time respectively~\cite{HurtadoMRS04}. A $2$-approximation algorithm is known for \textsf{APRBS}~\cite{CalinescuDKW05} by casting the separation problem as a special case of the rectangle stabbing problem\footnote{This is based on the idea that lines separating $R$ from $B$ must stab all rectangles formed by red and blue points at the corners. The parameterized version of rectangle stabbing is known to be intractable.}. We note that the $2$-approximation algorithm and \textsf{APX}-hardness applies even to the separation of monochromatic point sets (where the goal is to separate \emph{all} points from each other) and this version problem is also known to admit a approximation (OPT $\log$ OPT)-approximation~\cite{Har-PeledJ20}.


\paragraph*{Our Contributions}

We first address a question raised in the discussions from~\cite{CalinescuDKW05}: \emph{Do special cases, e.g., points in convex position, admit better approximation ratios or even exact solutions?} We make partial progress on this question by answering it in the affirmative when the input points lie on a circle (which would be a special case\footnote{We speculate that these ideas would also be relevant for the more general scenario of points in convex position. While the algorithm for \textsf{RBS} in fact works as-is for points in convex position, the details for the axis-parallel variant are less obvious.} of points in convex position). In particular, we show that when points lie on a circle, both \textsf{RBS} and \textsf{APRBS} admit exact polynomial-time algorithms. Interestingly, the \textsf{RBS} problem is significantly simpler in this special case compared to its axis-parallel counterpart. For the latter, the size of the optimal solution is captured by a structural parameter of a graph that is naturally associated with the point set. Our proof of this fact is algorithmic and can be used to solve the associated computational question. 

Further, we introduce a natural variant of \textsf{APRBS}, which is the $(p,q)$-\textsf{APRBS} problem. Here, as before, we are given a set of red and blue points, and the question is if there is a set of at most $p$ horizontal lines and at most $q$ vertical lines that separate $R$ from $B$. We show that this problem is \textsf{W[2]}-hard when parameterized by $p$ alone.

The rest of the paper is organized as follows. We formally define the problems that we address in Section~\ref{sec:prelims} and focus on the case of points on a circle in Section~\ref{sec:circle} for both \textsf{RBS} and \textsf{APRBS}. We describe the \textsf{W[2]}-hardness result for $(p,q)$-\textsf{APRBS} parameterized by $p$ in Section~\ref{sec:whard}.



\section{Preliminaries}
\label{sec:prelims}

For positive integers $x$, $y$, let $[x]$ be the set of integers between $1$ and $x$, and $[x, y]$ the set of integers between $x$ and $y$. Given a set of points $R \cup B$ in the plane, $R$ is said to be separated from $B$ by a collection of lines $L$ if every straight-line segment with one endpoint in $R$ and the other one in $B$ is intersected by at least one line in $L$. We adopt the convention of requiring a ``strict'' separation, which is to say that no point in $R \cup B$ is on a separating line. We let $n := |R \cup B|$, $r = |R|$ and $b = |B|$. The computational problems that we study are the following. 

\textsc{Red-Blue Separation.} (\textsf{RBS}) Given a set $R$ of red points and a set $B$ of blue points in the plane and a positive integer $k$ as input, determine if there exists a set of at most $k$ lines that separate $R$ from $B$.

\textsc{Axis-Parallel Red-Blue Separation.} (\textsf{APRBS}) Given a set $R$ of red points and a set $B$ of blue points in the plane and a positive integer $k$ as input, determine if there exists a set of at most $k$ axis-parallel lines that separate $R$ from $B$.

(p,q)-\textsc{Axis-Parallel Red-Blue Separation.} ((p,q)-\textsf{APRBS}) Given a set $R$ of red points and a set $B$ of blue points in the plane and positive integers $p$ and $q$ as input, determine if there exists a set of at most $p$ horizontal and $q$ vertical lines that separate $R$ from $B$.

\section{Points on a Circle}
\label{sec:circle}

In this section, we focus on the special case when all the points lie on a circle $C$. Let $P = (R \cup B)$ denote a set of $n$ points on a circle, with $r$ red points and $b$ blue points. As usual, $R$ (respectively, $B$) denotes the set of red (respectively, blue) points. Without loss of generality, we assume that all points of $P$ lie on an unit circle centered at the origin. Fix an order $\sigma$ on $R \cup B$ based on the order of their appearance on the circle, starting at $(1,0)$ and moving around the circle counterclockwise. We let $r_i$ and $b_j$ denote the $i^{th}$ red and the $j^{th}$ blue point that we encounter in this order. For a point $p$ on the circle, we let $\mathsf{col}(p)$ denote the color of the point $p$.

We call a maximal sequence of monochromatic points in $\sigma$ a \emph{chunk}. Let $\mathcal{C}_P$ denote the set of chunks of $P$. In the order of their appearance on the circle, we denote the individual chunks by $C_1, \ldots, C_w$. The color of a chunk is the color of any point belonging to it. We refer to a chunk consisting of red (blue) points as a red (blue) chunk. We overload notation and let $\mathsf{col}: \mathcal{C}_P \rightarrow \{R,B\}$ be a function that returns the color of a chunk. The arc of $C$ starting at the last point on the $r^{th}$ chunk and the first point on the $(r+1)^{th}$ chunk is called a \emph{switch}. Let $\mathcal{S}_P$ denote the set of switches of $P$. Note that any instance with $w$ chunks has $w$ switches. We denote the switches by $S_1, \ldots, S_w$ in the order of their appearance on the circle. Also note that $w$ must always be even, and that there are $\frac{w}{2}$ red chunks and $\frac{w}{2}$ blue chunks. We say that a switch $S$ is \emph{stabbed} by a line $\ell$ if $\ell \cap S \neq \emptyset$. We first make the following observation. 

\begin{proposition}
    \label{prop-switches}
Let $P = (R \cup B)$ be a red-blue point set on a circle. If $L$ is a set of lines that separates $R$ from $B$, then every switch must be stabbed by some line in $L$. 
\end{proposition}

\begin{proof}
Suppose that there exist a switch $S_i \in \mathcal{S}_P$ that is not stabbed by any line from the set $L$. Note that $S_i$ separates the chunks $C_i$ and $C_{i+1}$. Without loss of generality, suppose col($C_i$)= $R$ and col($C_{i+1}$)=$B$. Since $S_i$ is not stabbed by any line from $L$, the last point of $C_i$ and first point of $C_{i+1}$ are not separated, which leads to the contradiction of the fact that $L$ separates $R$ from $B$. 
Therefore, every switch must be stabbed by some line in $L$.
\end{proof}

Based on Proposition~\ref{prop-switches}, we have that in an instance with $w$ switches, $\frac{w}{2}$ is a lower bound on the optimum, since any line can stab at most two switches. In the next subsection, we show that this can always be achieved by a set of general lines. For axis-parallel lines, we strengthen the lower bound further using an auxiliary graph structure on the switches, and demonstrate an algorithmic argument to match the stronger lower bound. 

\subsection{The Case of General Lines}

With arbitrary lines, our strategy is simple: we ``protect'' each monochromatic chunk of a fixed color with a single line passing through it's adjacent switches. In particular, consider any chunk $C_i$ such that $\mathsf{col}(C_i) = B$. Fix an arbitrary point $p_i$ in the switch\footnote{If $i = 1$, then we let $S_{-1} := S_w$.} $S_{i-1}$ and another point $q_i$ in the switch $S_i$. Let $\ell_i$ be the line passing through $p_i$ and $q_i$ and let $L := \{\ell_i~|~ \mathsf{col}(C_i) = B\}$. In other words, $L$ is the set of lines thus defined based on blue chunks. Note that there are $\frac{w}{2}$ lines in this solution, since we introduce one line for each blue chunk. Moreover, it is also easy to check that these lines separate $R$ from $B$, since every blue chunk belongs to a separate cell of this configuration. 


\subsection{The Case of Axis-Parallel Lines}

When we are restricted to axis-parallel lines in the solution, then the strategy described in the previous subsection would fail since the lines that are described need not be axis-parallel. A similar strategy does give us a simple $2$-approximation, which we describe informally. Observe that each monochromatic chunk can be protected by a ``wedge'' consisting of a pair of axis-parallel lines. Indeed, consider the points $p_i$ and $q_i$ defined as before, and let $R$ be the unique rectangle whose sides are axis-parallel and which has $p_i$ and $q_i$ as diagonally opposite corner points. Clearly, one of the other two corner points $r$ lies inside $C$. We can now choose the two axis-parallel lines that contain the edges of the rectangle which intersect at $r$, and we have a wedge-like structure that protects the chunk (depending on the length of the chunk, note that the points of the chunk may be distributed over multiple cells). This gives us a solution with $w$ lines, and is therefore a two-approximate solution.

We now demonstrate a stronger lower bound for the setting of axis-parallel lines. To this end, we introduce some terminology and define an auxiliary graph based on the point set $P$. We say that a pair of switches \emph{face each other} if there exists a horizontal or vertical line that stab both of them. A switch which faces at least one other switch is said to be \emph{nice}, a pair of switches facing each other is called a \emph{nice pair}, and a switch that is not nice is said to be \emph{isolated}. We define a graph based on $P$ that has a vertex for every switch, and an edge between every pair of vertices corresponding to switches that are nice pairs. Formally, for a red-blue point set $P = (R \cup B)$ with $w$ switches $S_1, \ldots, S_w$, we define the graph $G_P = (V_P,E_P)$ as follows: $V_P = \{v_j ~|~ 1 \leq j \leq w\} \mbox{ and } E_P = \{(v_i,v_j)~|~ (S_i,S_j) \mbox{ is a nice pair} \}$. 


Observe that every isolated switch of $P$ corresponds to an isolated vertex of $G_P$. Recall that a \emph{edge cover} of a graph $G$ is a set of edges such that every vertex of the graph is incident to at least one edge of the set. Note that a minimum-sized edge cover can be found by greedily extending a maximum matching of a graph $G$. We use the abbrevation MEC to refer to a minimum edge cover. Let $I_P \subseteq V_P$ be the set of isolated vertices of $G_P$ and let $H_P := G_P \setminus I_P$. We define $ \kappa(G_P) := |I_P| + \mathsf{MEC}(H_P),$ where $\mathsf{MEC}(G)$ denotes the size of a minimum edge cover of the graph $G$. Our first claim is that any instance $P = (R \cup B)$ of \textsf{APRBS} requires at least $\kappa(G_P)$ lines to separate $R$ from $B$. Next, we will show that this bound is tight. 

Before stating the claims formally, we make some remarks about the bound. Note that this coincides with the bound obtained as a consequence of Proposition~\ref{prop-switches} when $G_P$ has a perfect matching. Further, the bound is $w$ when $G_P$ is the empty graph, or equivalently, when every switch is an isolated switch. In this scenario, note that the approach described for the two-approximate solution will, in fact, yield an optimal solution. The intuition for the bound in the generic case is the association between lines and edges in a MEC: indeed, every edge $e$ in $G_P$ corresponds to a family of lines that stab switches corresponding to the endpoints of $e$.  Our goal is to show that we can pick one line corresponding to each edge in the MEC and one line for each isolated switch in such a way that we separate $R$ from $B$. However, it is easy to come up with examples where this does not happen, and indeed, the argument for the upper bound follows by making a bounded number of modifications to the set of lines that was proposed with guidance from the MEC. On the other hand, this association runs both ways, so we can recover subset of edges from any collection of lines separating $R$ from $B$, using that stab two switches. If a solution uses fewer than $\kappa(G_P)$ lines, the hope is that the edges recovered lead us to an edge cover that has fewer edges than the MEC, which would be a contradiction. We now formalize both sides of this argument. We begin with the lower bound.

\begin{lemma}
\label{lem:lowerbound-circle} Let $P = (R \cup B)$ be a red-blue point set on a circle. Let $L$ be a set of $k$ axis-parallel lines that separate $R$ from $B$. Then $k \geq \kappa(G_P)$. 
\end{lemma}

\begin{proof}
    Consider any solution $L$ that uses $k$ axis-parallel lines. By Proposition~\ref{prop-switches}, we know that every switch must be stabbed by some line from $L$. In particular, suppose there are $\alpha$ lines in $L$ that stab a pair of (nice) switches, and $\beta$ lines that stab one switch. Clearly, $\beta \geq |I_P|$, the number of isolated vertices in $G_P$.
    
     Let $X$ be the set of non-isolated vertices in $G_P$ which are not covered by the edges corresponding to the $\alpha$ lines stabbing pairs of nice switches. Now, note that the switches corresponding to these non-isolated vertices in $X$ must be stabbed by one of the $\beta$ lines. So, $\beta \geq |X|+|I_P|$. Recall that $H_P = G_P \setminus I_P$ and $\mathsf{MEC}(H_P)$ covers every non-isolated vertex of $G_P$. Observe that $\mathsf{MEC}(H_P) \leq \alpha + |X|$, since the edges corresponding to the $\alpha$ lines stabbing pairs of switches can be extended by a collection of $|X|$ edges, one each for each non-isolated vertex that is not accounted for so far, to obtain a MEC for $H_P$. Adding both the inequalities above, we get:
     $$\alpha + \beta + |X| \geq  |X|+|I_P|+ \mathsf{MEC}(H_P) $$ 
     $$ \Rightarrow \alpha + \beta  \geq  |I_P|+ \mathsf{MEC}(H_P)$$
     $$ \Rightarrow k \geq \kappa(G_P),  $$
     
     as desired.
\end{proof}
We now turn to the upper bound.

\begin{lemma}
\label{lem:upperbound-circle} Let $P = (R \cup B)$ be a red-blue point set on a circle. There exists a collection of at most $\kappa(G_P)$ lines that separate $R$ from $B$.
\end{lemma}

The proof of the upper bound is algorithmic, and we demonstrate it with a series of claims. To begin with, let $F_P \subseteq E_P$ be a MEC of $G_P$ and let $t := |I_P|$. We define a set of lines $L_0$ as follows. For every edge $e = (v_i,v_j) \in F_P$, let $\ell_e$ be an arbitrary axis-parallel line passing through the switches $S_i$ and $S_j$. For every isolated switch $S_r$, let $\ell_r$ be an arbitrary axis-parallel line stabbing $S_r$. Now define $L_0$ as the collection of all of these lines, i.e:

$$L_0 = \{\ell_e ~|~ e \in F_P\} \cup \{\ell_r ~|~ v_r \in I_P \}.$$

Note that $|L_0| = \kappa(G_P)$. If $L_0$ separates $R$ from $B$, then we are done. Otherwise, we will obtain another set of axis-parallel lines that ``dominates'' $L_0$ in that it has the same size as $L_0$, separates all pairs of points separated by $L_0$ and at least one additional pair. To formalize this, we introduce the notion of a strictly dominating solution. For a set of lines $L$, let $\mathsf{sep}(L) \subseteq R \times B$ denote the set of red-blue point pairs that are separated by $L$. Given two collections of axis-parallel lines $L$ and $L^\star$, we say that $L^\star$ \emph{strictly} dominates $L$ if $|L^\star| \leq |L|$ and $\mathsf{sep}(L) \subsetneq \mathsf{sep}(L^\star)$. We will now show that there exists a sequence of sets of axis-parallel lines $L_0, L_1, \ldots, L_g$ such that $L_i$ strictly dominates $L_{i-1}$ for all $1 \leq i \leq g$ and $L_g$ is separates $R$ from $B$. Note that the number of steps is bounded by $rb$. Throughout, we will maintain the invariant that every switch is stabbed by at least one line. Note that this is true, in particular, for $L_0$.

\begin{claim}
    Every switch is stabbed by at least one line from $L_0$.
\end{claim}

\begin{proof}
    If the switch, say $S_r$, is isolated, then it is stabbed by the line $l_r \in L_0$, corresponding to the isolated vertex $v_r \in G_P$. Any other switch $S_i$ corresponds to a vertex in $G$ which is the endpoint of at least one edge $e$ in the MEC and the switch is stabbed by the line $\ell_e$.  
\end{proof}


For a collection of axis-parallel lines $L$, we say that a cell of $L$ is \emph{corrupt} if it is non-monochromatic, that is, if it contains at least one red point and at least one blue point. Note that $L_0$ contains at least one corrupt cell, otherwise we would be done. We consider all the possible ways in which a cell can intersect the circle underlying our point set. 

\begin{claim}
    Let $\mathcal{R}$ be an axis-parallel rectangle and let $C$ be a circle centered at the origin. Then $\mathcal{R} \cap C$ is either empty or consists of at most four disjoint arcs of $C$. 
\end{claim}

\begin{proof}
    Suppose $\mathcal{R}$ is an axis-parallel rectangle that lies completely inside the circle $C$ such that it does not intersect the boundary of $C$. Then, since all the points P lie on the boundary of $C$, so, $\mathcal{R} \cap C$ is empty. Otherwise, for $\mathcal{R}$ to intersect $C$ non-trivially, at least one arc must be contained inside $\mathcal{R}$. 
    Now, if $\mathcal{R}$ is such that all of its four axis-parallel lines form a secant of the circle, then there can be at most four disjoint arcs of $C$, captured by the four right-angles of the rectangle.
    
    More explicitly, if exactly one line of $\mathcal{R}$ is a secant, then $\mathcal{R} \cap C$ consists of at most two disjoint arcs (wedged by the secant and the two perpendiculars to it). If two opposite lines of $\mathcal{R}$ are secants, then again there are two disjoint arcs, wedged between those two secants. If two adjacent lines are secants, then there can be three arcs wedged out at the three right angles corresponding to these two adjacent secants. And finally, if there are three secants, then the remaining fourth line is forced either to be a chord or a secant, in which case there are either two or four disjoint arcs respectively. In any case, there can never be more than four disjoint arcs contained in $\mathcal{R} \cap C$. 
\end{proof}

Next, we note that any corrupt cell must contain at least two disjoint arcs of the circle.

\begin{claim}
    Let $L$ be a set of lines that stabs every switch at least once, and let $\mathcal{R}$ be a corrupt cell of $L$. Then $\mathcal{R} \cap C$ contains at least two disjoint arcs of the circle $C$. Further, all points that appear together on any arc of $\mathcal{R} \cap C$ are of the same color. 
\end{claim}

\begin{proof}
    Since $\mathcal{R}$ is a corrupt cell, it must contain at least one red and one blue point. So $\mathcal{R} \cap C$ is non-empty and contains at least one arc of $C$. Now, suppose $\mathcal{R} \cap C$ contains exactly one arc, which must now contain at least one red and one blue point. Observe that this implies that there is at least one switch $S$ that will be completely contained in the single arc of $\mathcal{R} \cap C$. However, this in turn implies that $S$ is not stabbed by any line from $L$, violating the assumption that every switch is stabbed by at least one line. Therefore, if $\mathcal{R}$ is a corrupt cell in a set of lines that stab every switch, we conclude that $\mathcal{R} \cap C$ contains at least two disjoint arcs. The second part of the claim follows, again, from the fact that the cell cannot contain a switch entirely. 
\end{proof}
We say that a cell $\mathcal{R}$ is \emph{large} if $\mathcal{R} \cap C$ contains three or four disjoint arcs. We now note that any instance can have at most one large cell.

\begin{claim}
    Any set of lines $L$ has at most one large cell. 
\end{claim}

\begin{proof}
    We claim that every large cell must contain the center of the circle. If not, then the cell can not contain any line passing through the centre, forcing it to completely lie at either side of the diagonal. In that case, the cell has at most two secants (corresponding to the two opposite lines of the cell) and hence, can not contain more than two disjoint arcs, contradicting the fact that it is a large cell. Therefore, a large cell must contain the centre. If there are two or more large cells, then they must overlap at the centre. But clearly, two cells can never overlap, else they will break into three or more smaller cells. Hence, any set of lines $L$ has at most one large cell.
\end{proof}

We are now ready to describe the high-level strategy for obtaining a sequence of strictly dominating solutions. It turns out that if a corrupt cell consists of exactly two disjoint arcs, then depending on the ``location'' of the cell, there is a simple strategy that allows us to clean up the cell by flipping the orientation of one of the lines in the solution. In particular, and informally speaking, the strategy works for corrupt cells that are ``above'' (``below'') the origin if all cells above it are monochromatic, or corrupt cells ``to the left'' (``to the right'') the origin if all cells before (after) it are monochromatic. This gives us a natural sweeping strategy to clean up corrupt cells from four directions, while potentially getting stuck at a large cell ``at the center''. When the large cell is the only corrupt one, it turns out that there are a fixed number of configurations it can have when considered along with its surrounding cells, and for each of these cases, we suggest an explicit strategy to clean up the large cell to arrive at a solution with no corrupt cells at all. We now formalize this argument.

\begin{center}
    \begin{figure}[htbp]
        \includegraphics{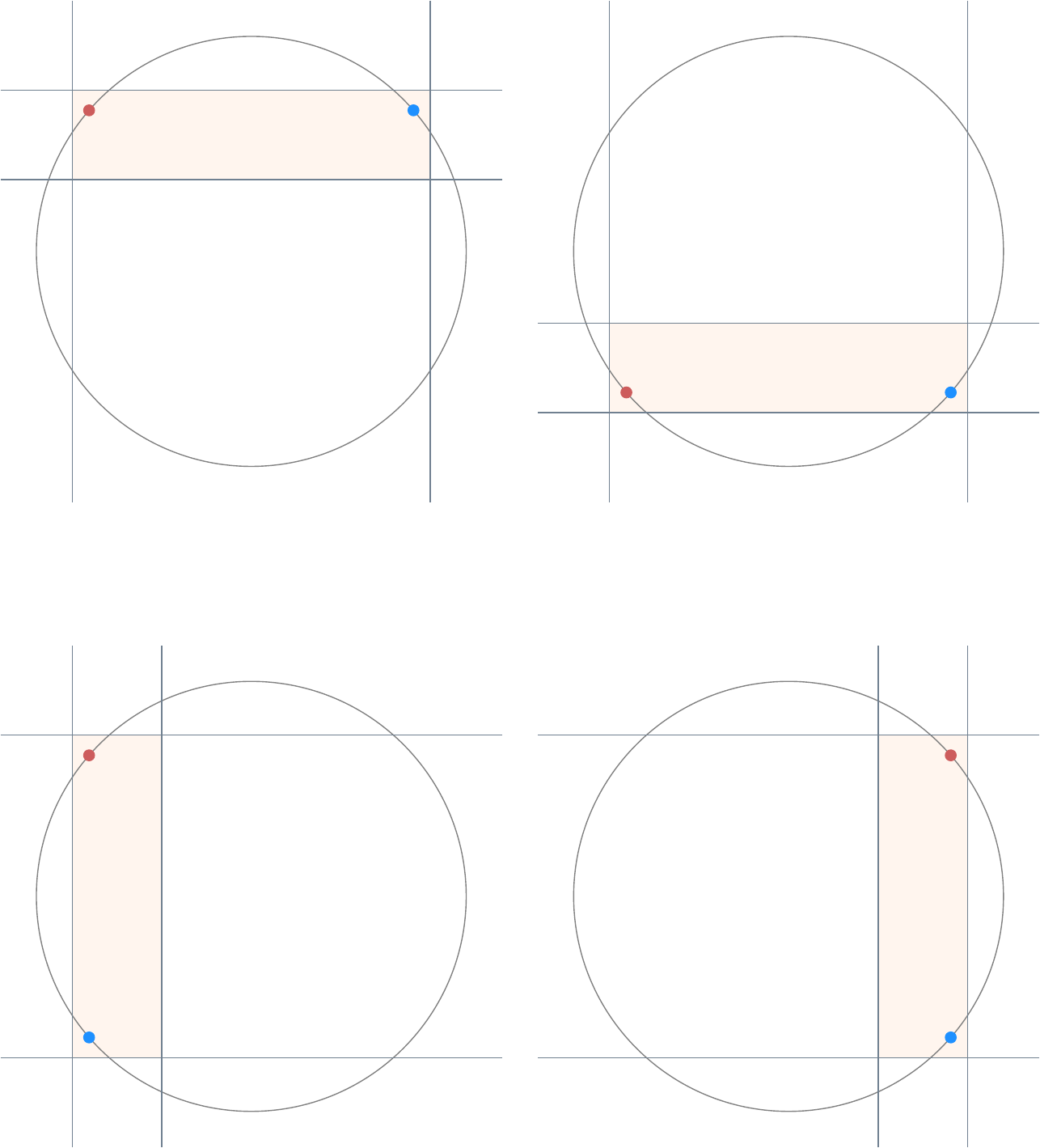}
    \caption{This figure demonstrates the possibilities for corrupt cells that intercept the underlying circle at exactly two arcs. Note that the choice of red and blue points are just specific examples, and the opposite scenarios may also arise.}
    \label{fig:twoarcstypes}
    
    \end{figure}
    \end{center}

Let $L$ denote the current solution: to begin with, $L$ is $L_0$, and we describe a process to obtain a solution $L'$ that strictly dominates $L$ if $L$ is not already a valid solution. Note that $L$ divides the plane into vertical and horizontal strips, which we will refer to as the rows and columns of the solution. Also, we call a cell of this configuration \emph{empty} if it does not contain any points of $P$. We first focus on corrupt cells that are \emph{not} large. Consider a cell $\mathcal{R}$ whose intersection with $C$ consists of exactly two disjoint arcs, say $A_1$ and $A_2$. Note that $A_1$ and $A_2$ lie in distinct quadrants. We call $\mathcal{R}$ a horizontal cell if these arcs lie in the first and second or the third and fourth quadrants; and we call $\mathcal{R}$ a vertical cell if these arcs lie in the first and fourth or the second and third quadrants. Note that the remaining possibilities do not arise with cells that are not large. We refer the reader to Figure~\ref{fig:twoarcstypes} for a visual representation of these cases. 

Consider the corrupt horizontal cell whose center has the largest $y$-coordinate in absolute value. This is either the top-most corrupt cell above the $x$-axis (Case 1) or the bottom-most corrupt cell below the $x$-axis (Case 2). If there are no corrupt horizontal cells, then consider the corrupt vertical cell whose center has the largest $x$-coordinate in absolute value. This is either the left-most corrupt cell to the left of the $y$-axis (Case 3) or the right-most corrupt cell to the right of the $y$-axis (Case 4). 

Let us consider Case 1. Here, observe that any row above the row containing the cell $\mathcal{R}$ consists of at most one non-empty cell and that all such cells are monochromatic by the choice of $\mathcal{R}$. Now, if the cell above $\mathcal{R}$ is monochromatic red and the arc in the first (second) quadrant consists of red points, then the top line of $\mathcal{R}$ can be flipped to a vertical line about the top-left (top-right) corner of the cell $\mathcal{R}$. It is easy to check that this solution strictly dominates $L$. The case when the cell above $\mathcal{R}$ is monochromatic blue can be argued similarly. We refer the reader to Figure~\ref{fig:twoarcsfix} for a demonstration of the switching strategies in these scenarios. 

\begin{center}
    \begin{figure}[htbp]
        \includegraphics{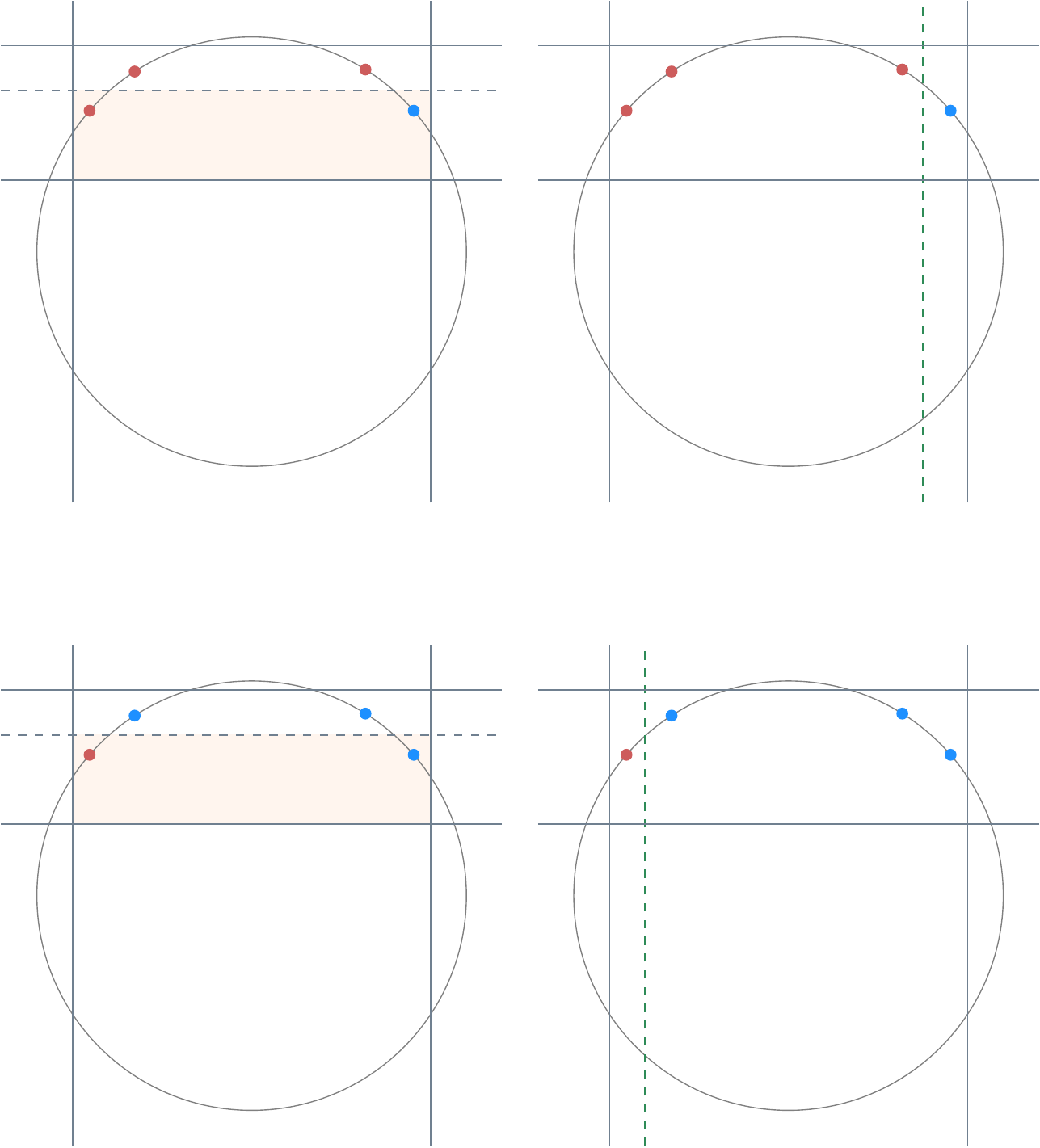} 
    \caption{This figure demonstrates the strategy for fixing a corrupt cell from Case 1. The shaded boxes on the figures on the left are the corrupt cells under consideration. The dotted line on the left is the line to be flipped, and the green dashed line on the right shows the line in its new flipped state: the choice of the flip depends on the state of the cell above and the colors of the points on the two arcs in the corrupt cell. The key observation is that if the cell above is monochromatic then there is at least one arc from the corrupt cell that points from the cell above can be safely ``exposed'' to, and the flip takes advantage of this while simultaneously fixing the corrupt cell. The other cases can be handled analogously.}
    \label{fig:twoarcsfix}
    
    \end{figure}
\end{center}
    
Case 2 is similar to Case 1 except that we argue relative to the cells below $\mathcal{R}$ rather than above it. In Cases 3 and 4, we flip vertical lines to a horizontal orientation, and the argument is based on monochromatic cells that lie to the left and right of $\mathcal{R}$, respectively. All the details are analogous to the case that we have discussed. Therefore, as long as the current solution has a corrupt cell that is not large, this discussion enables us to find a strictly dominating solution. 

Now, the only case that remains is the situation where we have exactly one corrupt cell which is large. For a large cell we have four surrounding monochromatic or empty cells. The three or four arcs contained inside the large cell may also have red or blue points in different configurations. It turns out that each of these cases admits a new solution which makes all cells monochromatic. This can be established by inspection, and we refer the interested reader to the supplementary material that goes over all the individual cases\footnote{All the cases can be downloaded from~\url{http://neeldhara.com/files/aprbs-cases.pdf}.}. Meanwhile, we refer the reader to Figure~\ref{fig:largebox} for a stereotypical case and the corresponding strategy. Based on this discussion, we conclude with the formal statement of the main result of this section. 

\begin{theorem}
If $P$ is a set of red and blue points on a circle, then \textsf{RBS} and \textsf{APRBS} can be solved in polynomial time. 
\end{theorem}

\begin{center}
    \begin{figure}[htbp]
        \centering
        \includegraphics{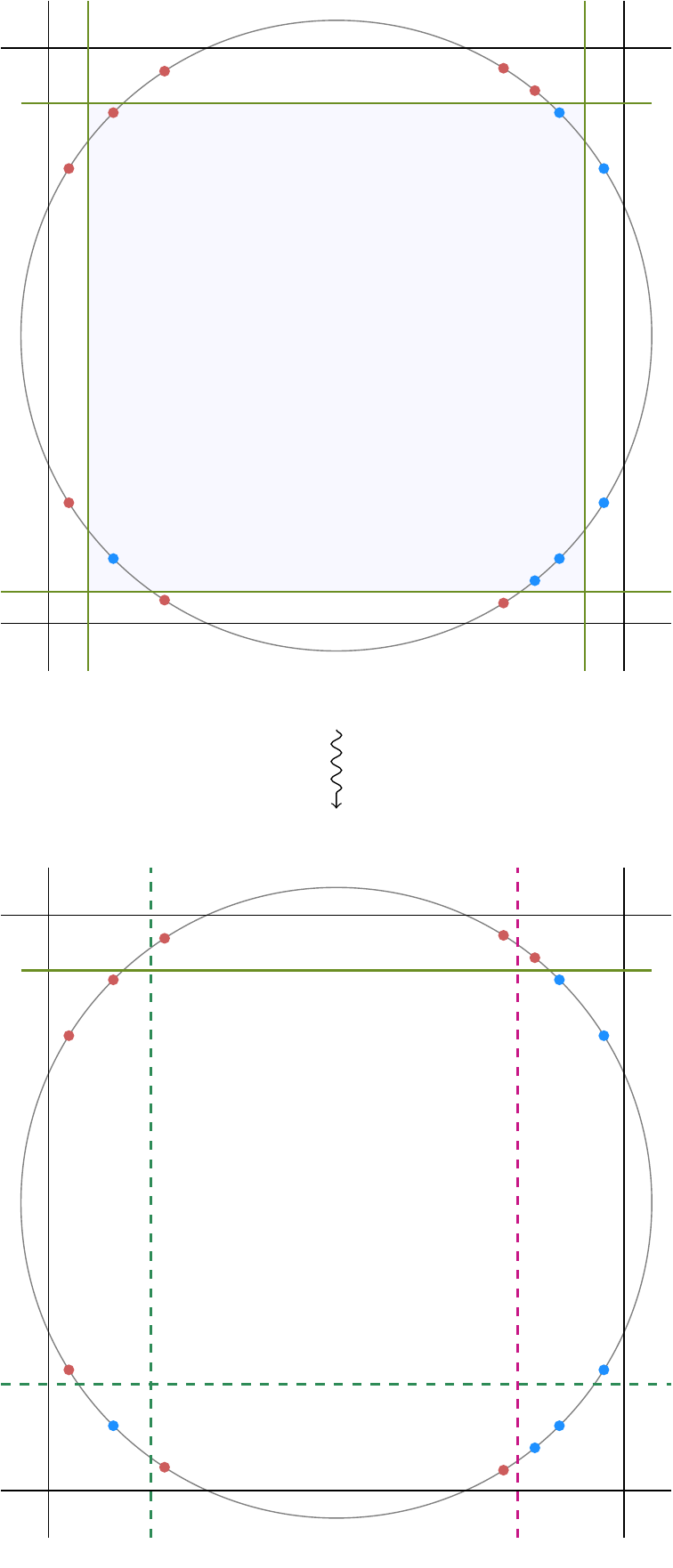} 
        \caption{This figure demonstrates demonstrates a typical scenario with a large corrupt cell. In this particular example, the rightmost line of the large corrupt cell is shifted inwards till it merges the two blue chunks on the top-right and bottom-right of the large cell with the monochromatic blue cell on the right. The extent to which this line has to be moved is determined based on the lengths of the chunks. After this, at least one of the top or bottom lines of the large cell will become partially redundant in that it can be removed without any damage in the first and fourth quadrants. Also, because the top-left chunk of the large cell has points whose color is the same as the adjoining cells on the top and the left, it is safe to flip the left boundary line of the large cell. Together with the redundant line, these two lines can safely form a ``wedge'' for the blue chunk on the bottom right corner of the large cell.}
        \label{fig:largebox}
        \end{figure}
    \end{center}

\section{W-hardness of (p,q)-Separation}
\label{sec:whard}

In this section, we focus on the $(p,q)$-\textsf{APRBS} problem. Before describing our result, we briefly comment on the relationship between this problem parameterized by only the budget for horizontal lines (or vertical lines, by symmetry) and \textsf{APRBS} parameterized by the size of the entire solution. If \textsf{APRBS} had turned out to be W[1]-hard or W[2]-hard parameterized by $k$, then it would imply that $(p,q)$-\textsf{APRBS} is unlikely to be \textsf{FPT} parameterized by either $p$ or $q$, since such an algorithm can be used as a black box to resolve the former question with only a polynomial overhead (guess $p,q$ such that $p+q = k$). On the other hand, if $(p,q)$-\textsf{APRBS} turns out to be \textsf{FPT} parameterized by either $p$ or $q$, then this would imply that \textsf{APRBS} is also \textsf{FPT} for the same reason. We show that $(p,q)$-\textsf{APRBS} is W[2]-hard when parameterized by $p$, the number of horizontal lines used in the solution. Therefore, our observation here establishes the hardness of the problem for a smaller parameter, and it does not have any direct implications for \textsf{APRBS}. Our result also is also not implied by what is known about \textsf{APRBS}, since it turns out that the problem is \textsf{FPT} parameterized by $k$.

We reduce from the \textsc{Colorful Red-Blue Dominating Set} (\textsf{C-RBDS}) problem,which is defined as follows. The input is a bipartite graph $G = (R,B)$ along with a partition of $R$ into $k$ parts $R_1 \uplus \cdots \uplus R_k$. The question is if there exists a subset $S \subseteq R$ such that $|R_i \cap S| = 1$ for all $1 \leq i \leq k$ and that dominates every vertex in $B$; in other words, for all $v \in B$, there exists a $u \in S$ such that $(u,v) \in E(G)$. Such a set is called a colorful red-blue dominating set for the graph $G$. This problem is well-known to be W[2]-hard when parameterized by $k$~\cite{CyganEtAl}.  Our reduction is inspired by the reduction from SAT used to show the hardness of the problem of separating $n$ points from each other~\cite{CalinescuDKW05}. One aspect that is specific to our setting is ensuring that the budget for lines in one orientation is controlled as a function of the parameter. 

\begin{theorem}$(p,q)$-\textsf{APRBS} is W[2]-hard when parameterized by $p$.
\end{theorem}

\begin{proof}
Let $G = (R = R_1 \uplus \cdots \uplus R_k, B); k)$ be an instance of \textsf{C-RBDS}. Without loss of generality, we assume that every vertex $v \in B$ has the same degree $d$ and that $d$ is even\footnote{When this is not the case, let $\Delta := \max_{v \in B}\{d(v)\}$. We may introduce an additional ``dummy color'' class $R_0$ with a forced dummy vertex (for example, by adding a $d$-star whose center is in $B$ and whose leaves are in $R_0$), and for every vertex $v \in B$ we may introduce $\Delta - d(v)$ new pendant red neighbors of $v$ in $R_0$. If $d$ happens to be odd, use $\Delta+1$ in this process instead of $\Delta$ to ensure that $d$ is even.}. We may also assume that all $R_i$'s have the same number of vertices (padding $R_i$ with $\max_{j=1}^k\{|R_j|\} - |R_i|$ dummy isolated vertices if required). We let $|R_1| = \cdots = |R_k| = m$ and $n := |B|$. We also assume that $k$ is even, again without loss of generality. Finally, we impose an arbitrary but fixed ordering on each $R_i$ and on the sets $N(v)$ for every $v \in B$.  

It will be convenient to think of the point set of the reduced instance as lying within a sufficiently large bounding box, say $\mathcal{B}$. To describe the placement of the points, we impose an uniform $(k+2) \times (n+1)$ grid on $\mathcal{B}$, which divides $\mathcal{B}$ into $k+2$ horizontal regions $H_0, H_1, \ldots, H_k, H_{k+1}$ (labeled from bottom to top) and $(n+1)$ vertical regions $V_0, V_1, \ldots, V_n$ (labeled from left to right) which we call \emph{tracks}. Each horizontal track $H_i$ for $i \in [k]$ is divided further into $m+2$ horizontal strips and each vertical track $V_j$ for $j \in [n]$ is divided further into $2d$ vertical strips. Within a horizontal track, the first and last horizontal strips are called \emph{buffer zones}. Further, when we refer to the $p^{th}$ horizontal strip within any horizontal track, the buffer zones are not counted. 

For $i \in [k]$, $j \in [n]$, $\alpha \in [m]$, and $\beta \in [2d]$, we refer to the intersection of the $\alpha^{th}$ horizonal strip in $H_i$ and the $\beta^{th}$ vertical strip in $V_j$ as $Z_{ij}[\alpha,\beta]$. We note that two points that share the same $x$-coordinate ($y$-coordinate) have to be separated by a vertical (horizontal) line. We now describe three sets of points that we need to add: the first will lead us to a choice of a vertex from each $R_i$, the second set encodes the structure of the graph, and the third set ensures that the chosen set is indeed a dominating set by forcing the use of a budget in a certain way. 

\textbf{Selectors.} Consider the first vertical track. Here, for any even (odd) $i \in [k]$, we add a red (blue) point to the top buffer zone and a blue (red) point to the bottom buffer zone of the $i^{th}$ track. These $2k$ points are called the \emph{selectors}. We ensure that all selectors have the same $x$-coordinate. Intuitively, the selector points ensure that any valid solution is required to use at least one horizontal line drawn in each of the $k$ horizontal tracks --- and the budget will eventually ensure that any valid solution uses exactly one. Which horizontal strip these lines end up in will act as a signal for our choice of vertices in the dominating set in the reverse direction. 

\textbf{Functional Points.} Next, consider any vertex $v_j \in B$. For every $u \in N(B)$, we add a pair of red and blue points in $Z_{ij}[\alpha,2\beta]$ if $u$ is the $\alpha^{th}$ vertex of $R_i$ and is the $\beta^{th}$ neighbor of $v_j$. These points are added to the bottom-left and top-right corners of the box. If $\beta$ is odd (even)\footnote{The organization of colors based on the parity of the columns in the case of functional pairs and rows in the case of selectors is to ensure that there are no additional separation requirements other than the ones that we desire to encode.}, then the bottom-left corner gets a blue (red) point and the opposite corner gets the red (blue) point. These pairs of points will be referred to as the \emph{functional pairs}. The functional pairs encode the structure of the graph, and we would like to ensure that the responsibility of separating at least one functional pair in each vertical track falls on a horizontal line used to separate the selectors. We force this by choosing an appropriately small budget for vertical lines, which ensures that not all separations can be accounted for using vertical lines. However, we still need to control how the vertical budget is utilized across different tracks. To this end, we introduce the final piece of our construction, which is a special gadget that forces the use of a certain number of vertical lines in each vertical track. 

\textbf{Guards.} In the horizontal track $H_0$, we place $d$ points, all with the same choice of $y$ coordinate which is arbitrary but fixed. Within the $j^{th}$ vertical track, $x$-coordinate of the $r^{th}$ point is chosen so that the point lies in the middle of the $(2r)^{th}$ vertical strip of $V_j$. The color of the first vertex in the track $V_i$ is blue if $i$ is odd and red if $i$ is even. This ensures that for $2 \leq i \leq n$, the first point in the $i^{th}$ track has the same color as the last point of the $(i-1)^{th}$ track. The colors of the remaining points are chosen so that consecutive points within the same track have distinct colors. Equivalently, the $s^{th}$ guard vertex in the $i^{th}$ track is blue (red) if $s$ and $i$ are both odd (even), and is red (blue) if $s$ is odd (even) and $i$ is even (odd). We refer to these points as \emph{guards}. We briefly discuss the role of the guard vertices: we note that the guards can be separated from each other by $(d-1)$ vertical lines, and since all guards have the same $y$-coordinate, this is the only way to separate them. However, there is no set of $(d-1)$ lines that can separate all the guards \emph{and} all the associated functional pairs in any vertical track. The budget for the vertical lines will be such that we can only afford to separate the guards as we are required to do, and we will be forced to separate at least one functional pair using a horizontal line, which will essentially ensure that the selectors have chosen vertices corresponding to a dominating set. 

We let $p := k+2$ and $q := (d-1)n + 1$. We also add three pairs of ``enforcer'' points as shown in Figure~\ref{fig:schematic}. This completes the description of the construction. We now turn to the argument for equivalence. 

\paragraph*{Forward Direction.} Let a colorful red-blue dominating set $S$ be given. Recall that $S$ contains exactly one vertex from each $R_i$ and every $v \in B$ has a neighbor in $S$. We now propose the following solution for the reduced instance. First, choose any three cannonical lines that separate the enforcer pairs as depicted by the dashed lines in Figure~\ref{fig:schematic}. We refer to this collection of lines as the \emph{fence}. Further, for $i \in [k]$, let $f(i)$ be the index of the vertex from $R_i$ in $S$ --- in other words, $S$ picks the $f(i)^{th}$ vertex from $R_i$ for all $i \in [k]$. Then, choose an arbitrary line in the $f(i)^{th}$ horizontal strip in the $i^{th}$ horizontal track and add it to the solution. We call these the \emph{signal} lines. Now, consider the $j^{th}$ vertical track. Let $g(j)$ be the smallest number such that the $g(j)^{th}$ neighbor of $v_j$ belongs to $S$ --- in other words, $g(j)$ is the smallest index among vertices of $N(v_j) \cap S$ with respect to the fixed order on $N(v_j)$. Note that $g(j)$ is well-defined for all $j \in [n]$ because $S$ is a dominating set. Now we introduce the following lines, which we refer to as the \emph{defender} lines:

\begin{itemize}
    \item for all $1 \leq r < g(j)$, we choose a vertical line in the $(2r)^{th}$ vertical strip whose $x$-coordinate is to the \emph{right} of the guard vertex in the strip, and
    \item for all $g(j) < r \leq d$, we choose a vertical line in the $(2r)^{th}$ vertical strip whose $x$-coordinate is to the \emph{left} of the guard vertex in the strip. 
\end{itemize}

This solution clearly uses $(k+2)$ horizontal lines (one signal line for each $H_i$, $i \in [k]$, and two lines based on the enforcer pairs) and $(d-1)n + 1$ vertical lines (we employ $(d-1)$ lines in each $V_j$, $j \in [n]$, and one line based on the enforcer pair), as desired. Now, we claim that these lines separate $R$ from $B$ by the following case analysis. 

\begin{itemize}
    \item \textbf{Guards.} The fence separates the guards from all non-guard points. Observe that the $(d-1)$ defender lines are chosen so that there is a line between every consecutive pair of guard vertices. The last guard vertex in any vertical strip has the same color as the first guard vertex in the adjacent vertical strip, so it is clear that the guard vertices belong to monochromatic cells.

    \item \textbf{Selectors.} The fence separates the selectors all non-selector points. Between any pair of selector points in a fixed horizontal track, the signal line chosen from the track separates said pair. Note that the last (or highest) selector point of $H_i$ and the first (or lowest) selector point of $H_{i+1}$ have the same color (for any $1 \leq i < k$). Finally, the bottom enforcer point that lies in $H_1$ and the first selector point of $H_1$ and the top enforcer point that lies in $H_k$ and the last selector point of $H_k$ are all red (since $k$ is even). Therefore, all selector points lie in monochromatic cells. 
    
    \item \textbf{Functional Pairs.} The fence separates the functional pairs from the selectors and the guards. Further, note that it is sufficient to argue that every functional pair is separated, since the points from different pairs are either in non-consecutive vertical strips and are separated anyway, or they are in adjacent vertical strips, in which case they have the same color by construction. Turning to the functional pairs themselves, it is easy to verify that if they are not separated by the defender lines, then they are separated by a signal line. We remark that pairs are separated ``doubly'' by both signal lines as well as defender lines --- this happens in vertical tracks corresponding to vertices that have more than one neighbor in $S$.
\end{itemize}

\paragraph*{Reverse Direction.} Let $L$ be a solution that uses $p$ horizontal lines and $q$ vertical lines. To begin with, note that the fence described in the forward direction is forced in the reverse direction by the enforcer points. This leaves us with a budget of $k$ horizontal lines and $(d-1)n$ vertical lines. Since no horizontal line can separate any pair of guard points, we know that every vertical track requires $(d-1)$ lines to separate the guard points that lie in it. Similarly, we also see that every horizontal track requires a horizontal line to separate the selector pair contained in the track, which would be impossible to separate using any vertical line. Because of the budget available at this point, we conclude that any valid solution consists of exactly one line in each horizontal strip $H_i$ for $i \in [k]$ and exactly $(d-1)$ vertical lines in each vertical strip $V_j$ for $j \in [n]$. We refer to these lines as signal and defender lines, respectively. 

Consider the signal line of the solution that lies in $H_i$. If it lies in the $\alpha^{th}$ strip of the track, we choose the $\alpha^{th}$ vertex from $R_i$ in our dominating set. Let $S$ denote this collection of $k$ vertices corresponding to the signal lines. We claim that every vertex $v_j \in B$ has a neighbor in $S$. Indeed, suppose not. Then consider the vertical track $V_j$. This track has $d$ functional pairs. The line segment joining each of these functional pairs does not intersect any of the signal lines (indeed, if it did, then the vertex corresponding to the signal line would be a neighbor of $v_j$ by construction). On the other hand, the track $V_j$ contains only $(d-1)$ defender lines, and it is easy to see that no vertical line can separate more than one functional pair (since each functional pair belongs to a distinct vertical strip), and therefore, $(d-1)$ vertical lines are not enough to separate $d$ functional pairs. This contradicts the assumption that $L$ separates $R$ from $B$, and also concludes our argument in the reverse direction. 
\end{proof}

\begin{figure}[htbp]
    \includegraphics{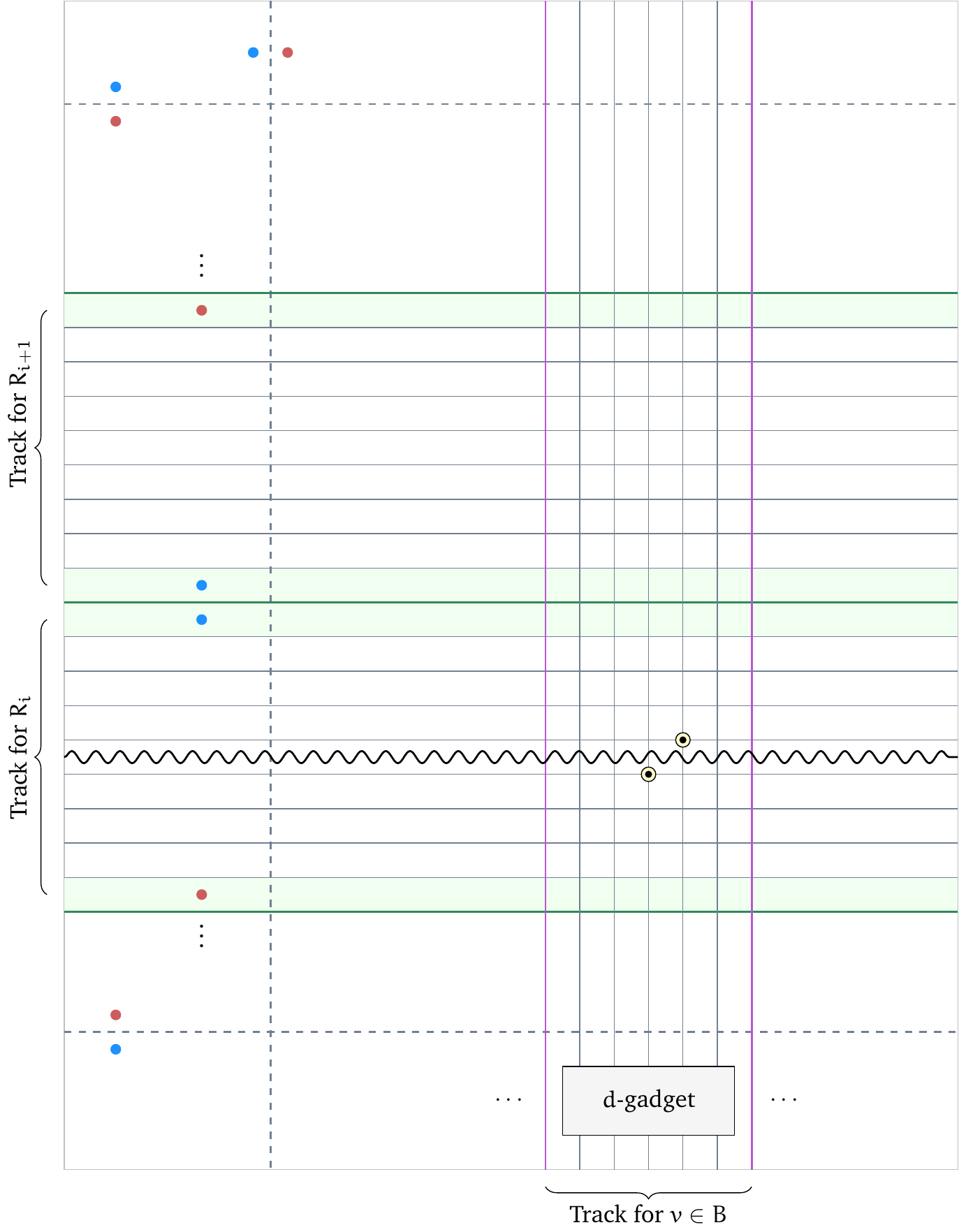}

    \caption{This figure shows the overall layout of the reduced instance demonstrating the tracks and strips as described. The colored strips are the buffer zones. For each $v \in B$, we introduce $d$ pairs of red-blue points at heights corresponding to the strip representing neighbors of $v$ (a possible pair is given by the highlighted points in this picture). Based on the budget for vertical lines and the structure of the $d$-gadget placed at the bottom of each zone corresponding to a blue vertex, we force that at least one of these pairs must be separated by a horizontal line lying in a strip corresponding to one of the neighbors of $v$ (represented in this figure by the wavy line), which will eventually correspond to the desired dominating set.}
    \label{fig:schematic}
    
\end{figure}



\section{Concluding Remarks}
We showed that \textsf{RBS} and \textsf{APRBS} are polynomial-time solvable when points lie on a circle. Further, we introduced a natural variant that separates out the budget for horizontal and vertical lines in the axis-parallel variant, and demonstrated that $(p,q)$-\textsf{APRBS} is \textsf{W[2]}-hard when parameterized by $p$. The most natural question in the context of the discussion about special classes of input is if the algorithm for the case of points on a circle can be generalized to points in convex position. We conjecture that this should be possible by a suitable adaptation of our arguments here. In the general setting, since \textsf{APRBS} is \textsf{FPT} when parameterized by $k$~\cite{KratschMMPS2020}, the question of whether the problem admits a polynomial kernel would be natural to explore further. Our \textsf{W[1]}-hardness reduction for $(p,q)$-\textsf{APRBS} may provide some starting points towards an answer in the negative --- in its present form the parameter $k$ of the reduced instance depends on $k,d,$ and $n$. \textsf{APRBS} would not admit a polynomial kernel (under standard complexity-theoretic assumptions) if this dependence can be reduced to $k$ and $n$ only~\cite{DomLS14}.


\bibliography{references}

\appendix

\end{document}